\documentclass[12pt]{article}

\usepackage{epsfig}
\usepackage{amssymb}
\usepackage{amsfonts}

\usepackage{color}

\def\be{\begin{equation}}
\def\ee{\end{equation}}
\def\ba{\begin{array}{c}}
\def\ea{\end{array}}
\def\p{\partial}
\def\ben{$$}
\def\een{$$}

\newcommand{\bea}{\begin{eqnarray}}
\newcommand{\eea}{\end{eqnarray}}

\newcommand{\bbr}{\br\!\br}
\newcommand{\kkt}{\kt\!\kt}
\newcommand{\pbr}{\prec\!}
\newcommand{\pkt}{\!\succ\,\,}
\newcommand{\kt}{\rangle}
\newcommand{\br}{\langle}

\newtheorem{thm}{Theorem}

\newtheorem{prop}[thm]{Proposition}

\newenvironment{proof}{\noindent {\bf Proof}}{\hfill$\square$\vspace{3mm}\endtrivlist}

\begin{document}

\titlepage


 \begin{center}{\Large \bf

Anisotropy and asymptotic degeneracy of the physical-Hilbert-space
inner-product metrics in an exactly solvable crypto-unitary quantum model

  }\end{center}


 \begin{center}

\vspace{8mm}

  {\bf Miloslav Znojil} $^{1,2,3}$

\end{center}

\vspace{8mm}

  $^{1}$
 {The Czech Academy of Sciences,
 Nuclear Physics Institute,
 Hlavn\'{\i} 130,
250 68 \v{R}e\v{z}, Czech Republic, {e-mail: znojil@ujf.cas.cz}}


 $^{2}$
 {Department of Physics, Faculty of
Science, University of Hradec Kr\'{a}lov\'{e}, Rokitansk\'{e}ho 62,
50003 Hradec Kr\'{a}lov\'{e},
 Czech Republic}

  $^{3}$
Institute of System Science, Durban University of Technology,
Durban, South Africa


\vspace{8mm}



\section*{Abstract}

The phenomenon of a sudden and complete loss of observability
is described in an exactly solvable ${\cal PT}-$symmetric
quantum model.
The collapse is
controlled by a time-dependent toy-model
Hamiltonian and by a time-dependent and
unitarity-guaranteeing Hilbert-space
metric.
At the time $\tau=0$
the process is initiated
in a conventional self-adjoint regime, and
it climaxes in the
degeneracy at $\tau=1$.
The model describes an
$N-$level quantum system where $N < \infty$ so that
both the Hamiltonian and the
metric
are $N$ by $N$ matrices
$H^{(N)}(\tau)$ and $\Theta^{(N)}(\tau)$, respectively.
Moreover, the initial
Hamiltonian $H^{(N)}(0)$ is postulated diagonal and the initial
metric is chosen trivial, $\Theta^{(N)}(0)= I$.
At all times $\tau \in (0,1)$
the metric is kept minimally anisotropic,
with
the evolution towards collapse
characterized by a steady increase of its anisotropy.
At the end of process with $\tau \to 1$
the metric becomes singular
(of rank one) and the end-point Hamiltonian $H^{(N)}(1)$ loses its
diagonalizability (i.e., it acquires
a canonical Jordan-block form).
At any time $\tau$ and dimension $N$, the best
insight in the evolution towards the ultimate quantum catastrophe
is provided by the formulae which give the spectrum of the metric
in closed form.

\section*{Keywords}

unitary quantum solvable toy model;
Hilbert-space anisotropy;
eigenvalues of metric;
exceptional-point collapse;

\newpage

\section{Introduction}

In many standard and routine applications of quantum theory
the evolution in time is
prescribed, in the so called Schr\"{o}dinger representation
\cite{Messiah}, by Schr\"{o}dinger equation
 \be
 {\rm i}\partial_{\tau}\,|\psi({\tau})\pkt
 = \mathfrak{h}\,|\psi({\tau})\pkt
 \label{1}
 \ee
where the state vector belongs to a physical Hilbert space of
conventional textbooks, $|\psi({\tau})\pkt \in {\cal H}^{(T)}$. The
Hamiltonian is usually assumed time-independent and self-adjoint in
${\cal H}^{(T)}$.
Often, the theory is realized in the physical
and, at the same time, user-friendly special
space ${\cal H}^{(T)}=L^2(\mathbb{R}^d)$
of square-integrable coordinate-dependent
functions $\psi(x,{\tau}) = \pbr
x|\psi({\tau})\pkt$ in $d$ dimensions.

Under these assumptions
the evolution is
unitary \cite{Stone} and Eq.~(\ref{1}) is formally solvable,
 \be
 |\psi({\tau})\pkt = e^{-{\rm i} \mathfrak{h} {\tau}}\,|\psi(0)\pkt\,.
 \label{22}
 \ee
The practical construction of the wave functions usually proceeds
via an approximate or exact diagonalization of
$\mathfrak{h}=\mathfrak{h}^\dagger$ \cite{Fluegge}.
The
description of the evolution remains equally routine for the
time-dependent Hamiltonians $\mathfrak{h}=\mathfrak{h}({\tau})$.
One may also move from the primary
Schr\"{o}dinger representation to
its equivalent Heisenberg-representation alternative.
Via a suitable
unitary operator
one then obtains
the Heisenberg-representation wave functions
which are required not to vary with time
\cite{Messiah}.

A more challenging theoretical as well as
conceptual scenario emerges
when the Heisenberg-representation-inspired preconditioning
of the wave-function ket-vector
 \be
 |\psi({\tau})\pkt = \Omega({\tau})\,|\psi\kt\,
   \label{33}
 \ee
is chosen invertible but non-unitary
\cite{Dyson}, i.e., such that
 \be
 \Omega^\dagger({\tau})\Omega({\tau})=\Theta({\tau})\neq I\,.
 \label{2}
 \ee
Then, product $\Theta({\tau})$ can be perceived as
playing the role of a
correct Hilbert-space metric in an
``amended'' physical Hilbert space ${\cal H}^{(A)}$
such that the ket-vectors
$|\psi\kt \in {\cal H}^{(A)}$
are simpler
in comparison with their more conventional
textbook avatars $|\psi\pkt \in {\cal H}^{(T)}$.

In most of the reviews of the state of art
(cf., e.g., \cite{Carl,ali,book})
the authors are trying to cover the
whole new terrain of the theory
on a rather abstract level. At the same time,
the simplification of the picture
caused by the change of paradigm
is usually mentioned just marginally,
as an assumption or a tacit wish rather than
as a rather difficult necessary condition of a
consequent practical
and constructive implementation of the formalism.

In our present paper we intend to pay more attention
to the postulates of solvability and simplification. Indeed,
these features of the quantum models of interest
represent one of the not often emphasized
keys to
applicability of the
whole non-unitary-preconditioning
idea behind the Dyson-inspired mapping (\ref{33}).

The presentation of our considerations and results will start
in section \ref{bm1} where we will summarize a few basic concepts
forming the theory. In subsequent section \ref{bm}
we will point out that in the context of
rigorous mathematics the theory itself is still
in the stage of formal development, characterized by the
existence of
a large number of open mathematical questions and challenges
(cf., e.g., \cite{Dieudonne,Trefethen,Viola}).
In this sense we decided to
circumvent some of these challenges
in the spirit of the words of warning
in reviews \cite{Dieudonne,Geyer}.
Thus,
we will just consider
a family of
sufficiently innocent-looking
benchmark models living in Hilbert spaces
of an arbitrary
{\em finite\,} dimension $N$.

One of the phenomenologically most relevant benefits
of such a choice of models will be discussed in
section \ref{uvodnik}. Its essence will be emphasized to lie in the
possibility of using the well known mathematical
ambiguity and flexibility of the inner-product metric $\Theta$
for the purposes of description of the quantum systems
in an arbitrarily small vicinity
of their singularities
representing certain forms of a
quantum catastrophe.

In section \ref{zavodnik}
we will modify the paradigm and
extend the
dynamical
framework of Schr\"{o}dinger
representation which is inherently stationary.
We will turn attention to a more explicit study of the
time-dependent aspects of our class of benchmark
(and exactly solvable) models.
In this section we will emphasize that
the dominant merit of our models lies in the
closed-form availability of
non-stationary metrics $\Theta=\Theta(\tau)$.

The details of the construction are made explicit
in sections \ref{prevodnik} and \ref{svod}
in which we present a mathematical core of our present message.
A basic
mathematical characteristic of our class of models
will be shown to lie
in the smoothness of the time-dependence of
the inner-product metrics $\Theta=\Theta(\tau)$
and, first of all, in the existence of these operators
for the times covering the whole interval
connecting, in one extreme, the Hermitian quantum mechanics
(characterized by the trivial and fully isotropic metric and
reached, in our units, at $\tau=0$)
with the other extreme of
a ``quantum catastrophic''
{\it alias\,}
``phase-transition'' {\it alias\,} ``fully degenerate''
collapse of the system in the $\tau\to 1$
limit in which the inner product
metric asymptotically and ultimately
degenerates and ceases to exist.

A few concluding remarks will be added in sections \ref{presummary}
and \ref{summary}.

\section{An outline of theory\label{bm1}}

The conceptual consistency of the non-unitary
Dyson's mapping (\ref{33}) is based
on the requirement of equivalence between
the evaluations of the old and new inner product,
 \be
 \pbr \psi_1|\psi_2 \pkt \ (= {\rm product\ in \ } {\cal H}^{(T)})\
 =
 \br \psi_1|\Theta({\tau})|\psi_2 \kt \
 (= {\rm product\ in \ } {\cal H}^{(A)})\,.
 \label{3}
 \ee
The main reason why the non-unitarity
$\Omega({\tau})\neq \Omega^\dagger({\tau})$
in Eq.~(\ref{33}) is
challenging is that the survival of the requirement
of equivalence of physics in
${\cal H}^{(T)}$ and ${\cal H}^{(A)}$
leads to the apparently counterintuitive definition (\ref{3}) of the
inner product in ${\cal H}^{(A)}$. Subsequently, it is fairly
difficult to resist the temptation of introducing another, third,
user-friendlier Hilbert space ${\cal H}^{(F)}$
in place of ${\cal H}^{(A)}$. In it, one re-accepts
the manifestly unphysical but simpler-to-use metric
$\Theta^{(F)}=I$ which only has to be remembered as
mathematically useful and preferable even though
manifestly unphysical.

In the case of the simplest, manifestly time-independent non-unitary
mappings $\Omega$ the trick (\ref{3}) and transition to the
``three-Hilbert-space'' (THS) representation of a given quantum
system proved particularly rewarding in
applications, say, in
nuclear physics \cite{Geyer,Jensen}.
The isospectrality of the mappings
 \be
 \mathfrak{h} \to H=\Omega^{-1}\,\mathfrak{h}\,\Omega
 \label{kasim}
 \ee
of the Hamiltonians
as induced by
Eq.~(\ref{33})
has been used there to facilitate
the practical variational
estimates of the bound state energies of certain
heavy nuclei.

Later on,
the THS formalism also found applications in
relativistic quantum field theory.
Emphasis has been redirected
to the study of systems exhibiting
the parity times time-reversal symmetry
{\it alias\,} ${\cal PT}-$symmetry
of the Hamiltonian
(cf., e.g., the dedicated
reviews \cite{Carl,ali} for an extensive
information and detailed discussion).

Quickly, the growth of the scope of the theory
has been also noticed and accepted in the other parts of physics like,
say,
experimental optics \cite{Christodoulides}.
Various unexpected consequences of the
generalization $\mathfrak{h} \to H$
have been found inspiring, especially when the researchers
managed to keep the evolution-generator
${\cal PT}-$symmetric.
This led to a new perception of Maxwell equations
(in the so called paraxial approximation)
and to the experiments using metamaterials with anomalous
refraction indices \cite{Makris,Makrisb,Makrisc}).

A purely quantum
theoretical as well as phenomenological appeal of the THS
approach reemerges when one opens the Pandora's box of
time-dependent problems \cite{timedep} -- \cite{Mousse}. First of
all, it is necessary to imagine that the second Hilbert space becomes
time-dependent in a way
mediated and carried by the
time-dependence of its
metric $\Theta=\Theta({\tau})$ \cite{NIP,Bishop}.

Thus, one has
${\cal H}^{(A)}={\cal
H}^{(A)}({\tau})$
and one
must replace the time-evolution Schr\"{o}dinger
Eq.~(\ref{1}) valid in ${\cal H}^{(T)}$ by its manifestly
non-Hermitian analogue
 \be
 {\rm i}\partial_{\tau}\,|\psi({\tau})\kt = G({\tau})\,|\psi({\tau})\kt\,
 \label{11}
 \ee
\cite{timedep,NIP}.
This version of evolution equation may still be considered and solved
in the unphysical but user-friendlier Hilbert space
${\cal H}^{(F)}$. In this case it is only necessary to keep
in mind that
the sophisticated non-Hermitian
generator of evolution
 \be
 G({\tau})= H(\tau|-\Sigma(\tau)
 \label{777}
 \ee
must be defined as composed of the observable Hamiltonian
 $$
 H({\tau})=\Omega^{-1}({\tau})\,
\mathfrak{h}({\tau})\,\Omega({\tau})
 $$
and of another operator
 $$\Sigma({\tau})={\rm
i}\,\Omega^{-1}({\tau})
 \,\left [\partial_{\tau}\Omega({\tau})\right ]\,
 $$
called quantum Coriolis force \cite{Mousse}.
Marginally, it may be added that both of these
components of the observable ``instant energy''
Hamiltonian $H(\tau)$
have, in general, complex spectra \cite{ITJP}.

Obviously, a consistent
version of the formalism requires a
cancelation of the non-Hermiticities
as carried by $G({\tau})$ and
$\Sigma({\tau})$.
In a way explained in brief note~\cite{jundemu}
the latter cancelation is still absent
in the systems with trivial
$G({\tau})=0$.
In papers, \cite{jundemu,FrFrith}, incidentally, such an
option
(simplifying
Schr\"{o}dingerian Eq.~(\ref{11})
and
reminding us
of
Heisenberg representation)
has been extended to cover also
the non-vanishing but still
simplified, viz.,
time-independent constant-operator
Schr\"{o}dingerian generators
$G({\tau})=G(0)\neq 0$.
Nevertheless, in the fully general version of the
formalism
one cannot rely upon
similar simplifications.

In particular,
the changes of the
physical inner product in ${\cal H}^{(A)}({\tau})$ need not be slow.
Hence, the Coriolis-force
operator $\Sigma({\tau})$
treated as a difference between $H(\tau)$ and $G(\tau)$
need not be small, either.
This might make the adiabatic approximation more or less
useless \cite{Bila,Bilab}. At the same time, we are
often {\em forced} to assume the validity of the adiabatic
approximation for practical purposes.
This is the situation in which one needs a methodical
guidance mediated, typically, by the exactly solvable examples.
Via their deeper analysis one can identify the dynamical regimes
in which
$\Sigma({\tau})$ can be kept small.

A family of models
characterized by an explicit, closed-form knowledge of
the relevant operators  will be
introduced and described in what follows, therefore.


\section{Benchmark model \label{bm}}

One of the most immediate consequences of
relations (\ref{3}) and (\ref{kasim})
is that the self-adjointness of $\mathfrak{h}(\tau)$
can equivalently be re-expressed
as the metric-dependent
quasi-Hermiticity \cite{Dieudonne} of $H(\tau)$
in ${\cal H}^{(F)}$,
 \be
 H^\dagger(\tau)\Theta(\tau)=\Theta(\tau)
 H(\tau)]\,.
 \ee
An analogous relation is also required to be satisfied
by any other candidate $\Lambda(\tau)$
for an observable.

In the light of
the Dieudonn\'{e}'s critical
analysis \cite{Dieudonne} the authors of review \cite{Geyer}
pointed out and strongly recommended that
all of the eligible candidates
for an observable (including, naturally, also the
energy-representing Hamiltonian)
should be, preferably, represented by operators which are,
in the friendly Hilbert space ${\cal H}^{(F)}$
of mathematical preference,
bounded.

In our present paper we followed the recommendation.

\subsection{Bounded-operator Hamiltonians}

In the context of prevailing model-building practice,
the constraint of boundedness
appeared over-restrictive \cite{Kretsch,Kretschb,Kretschc}.
Indeed, multiple popular elementary
forms of
Hamiltonians have the form of
a differential operator
which is
unbounded \cite{Bijan,Khare}.
In this light
and in a way recalled
at the beginning of section \ref{presummary} below
we decided to accept the constraint and to replace,
in one of our related papers~\cite{maximal},
the most common harmonic-oscillator ordinary-differential
Hamiltonian by its truncated and shifted,
$N-$dimensional diagonal-matrix equidistant-spectrum analogue
  \be
 H^{(N)}_{(LHO)}
 =\left [\begin {array}{cccc}
  -(N-1)&0&\ldots&0\\
 0& -(N-3)&\ddots&\vdots\\
 \vdots&\ddots&\ddots&0
 \\
 0&\ldots&0&+(N-1)
 \end {array}\right ]\,.
 \label{NyTSal}
 \ee
In parallel, the methodical and pedagogical role
of the most popular anharmonic-oscillator-like
Hermiticity-violating interactions \cite{Simon,Swanson}
was transferred to the off-diagonal elements of certain
real and, say, tridiagonal non-Hermitian $N$ by
$N$ multiparametric matrices
  \be
 H^{(N)}_{(AHO)}
 =\left [\begin {array}{cccccc}
  1-N&{}\,{g}_1&0&0&\ldots&0\\
 -{}\,g_1{}& 3-N&{}\,{g}_{2}{}&0&\ldots&0\\
 0&-{}\,g_{2}{}&5-N&\ddots&\ddots&\vdots
 \\
 0&0&-{}\,{g}_{3}{}&\ddots&{}\,{g}_{N-2}{}&0
 \\
 \vdots&\vdots&\ddots&\ddots&N-3&{}\,{g}_{N-1}{}\\
 0&0&\ldots&0&-{}\,g_{N-1}{}&N-1
 \end {array}\right ]\,.
 \label{NyTSbet}
 \ee
This enabled us to simplify the proofs of
the required reality of the bound-state spectra.
For models (\ref{NyTSbet}) we managed to reduce these proofs to
a virtually elementary spectral-continuity or spectral-inertia
argument, applicable at the not too large couplings ${g}_{j}{}$ at
least.

After an additional up-down symmetrization of the above
matrix, i.e., after the choice of
${g}_{N-1}{}={g}_{1}{}$ and ${g}_{N-2}{}={g}_{2}{}$, etc, we arrived
at the final form of our benchmark toy-model Hamiltonian
  \be
 H^{(N)}_{(PT)}{}
 =\left [\begin {array}{cccccc}
  1-N&{}\,{g}_1{}&0&0&\ldots&0\\
 -{}\,g_1{}& 3-N&{}\,{g}_{2}{}&0&\ldots&0\\
 0&-{}\,g_{2}{}&5-N&\ddots&\ddots&\vdots
 \\
 0&0&-{}\,{g}_{3}{}&\ddots&{}\,{g}_{2}{}&0
 \\
 \vdots&\vdots&\ddots&\ddots&N-3&{}\,{g}_{1}{}\\
 0&0&\ldots&0&-{}\,g_{1}{}&N-1
 \end {array}\right ]\,.
 \label{NyTSpt}
 \ee
After such a choice of the class of models we
also managed to parallel the phenomenologically relevant
parity times time reversal symmetry of multiple common
differential-operator toy-model
Hamiltonians by a formally analogous ${\cal
PT}-$symmetry as imposed upon our matrices (\ref{NyTSpt}).
This build-up of analogy merely required
the specification of ${\cal T}$ as the transposition plus sign
reversal. The parity-simulating indefinite square root of the unit
matrix appeared then represented by an antidiagonal $N$ by $N$
matrix ${\cal P}$ with non-vanishing
elements ${\cal P}_{j,N-j+1}=1$,
$j=1,2,\ldots,N$.

In our subsequent paper \cite{tridiagonal} we turned attention to
one of the methodically most welcome features of the ${\cal
PT}-$symmetric and $J=[N/2]-$parametric benchmark Hamiltonian
(\ref{NyTSpt}), viz., to the
availability of the
amazingly elementary geometric form of
boundary $\partial {\cal D}$ of the $J-$dimensional compact
domain ${\cal D}$ of the real parameters ${g}_{j}$ for which the
spectrum of $ H^{(N)}_{(PT)}$ remains real. For our models
(\ref{NyTSpt}) this boundary (or, in the language of physics, the
horizon of the bound-state stability of the system \cite{horizon})
has been shown to acquire, at any matrix dimension $N$, the same
generic geometric form of surface of a smoothly deformed
hypercube with protruded edges and vertices (cf. also
\cite{Siegldis} for more details).

\subsection{Fall in instability}

Initially, the family of our present toy models (\ref{NyTSpt})
has been developed with the purpose of having a tractable
sample of a quantum analogue of a classical
concept of an evolution singularity called,
in the Thom's popular terminology \cite{III,Zeemanb}, a ``catastrophe''.
This aim of study has been made explicit in our paper \cite{I}.
In place of our present
variable $\tau$
measuring the time during the fall of the system into its degenerate
singularity we used a different variable
$\lambda=1-\tau^2$, in terms of which some of the formulae
appeared simpler.

Indeed,
we revealed that  there
exists a certain specific ${\lambda}-$parametrization of the
couplings ${g}_{j}={g}_{j}({{\lambda}})$ in (\ref{NyTSpt}) such that

\begin{itemize}

\item
the two by two matrix $ H^{(2)}_{(PT)}({{\lambda}})$ appears useful
as a benchmark model of an {\em energy-bifurcation} scenario in
which
 \be
 E_0=-\sqrt{{\lambda}}, \ \ E_1= +\sqrt{{\lambda}}\,,
 \label{E5}
 \ee
i.e., in which the spectrum is real iff ${\lambda}\geq \lambda_0=0$
and in which the whole spectrum becomes completely degenerate iff
${\lambda}=0$ while it finally gets purely imaginary iff
${\lambda}<0$;

\item
the three by three matrix $ H^{(3)}_{(PT)}({{\lambda}})$
has been found to serve
as a benchmark model of a new, {\em energy-trifurcation} quantum
catastrophe in which
 \be
 E_0=-2\sqrt{{\lambda}}, \ \ E_1=0, \  \ E_2= +2\sqrt{{\lambda}}\,.
 \label{E6}
 \ee
Again, the spectrum proved completely degenerate iff ${\lambda}=0$.
Up to the exceptional, ${\lambda}-$independent real level
$E_{[N/2]}=0$ emerging at any odd $N$, the rest of the spectrum was
again purely real or imaginary iff ${\lambda}\geq 0$ or
${\lambda}<0$, respectively;

\item
the four by four matrix $ H^{(4)}_{(PT)}({{\lambda}})$ with spectrum
 \be
 E_0=-3\sqrt{{\lambda}}, \ \ E_1= -\sqrt{{\lambda}}, \ \ E_2=
 \sqrt{{\lambda}},\
  \ E_3= 3\sqrt{{\lambda}}
 \,,
 \label{E7}
 \ee
found then an analogous interpretation
of a benchmark quantum model admitting an {\em
energy-quadrifurcation}.

\end{itemize}

 \noindent
Analogous quantum-catastrophic (QC) features
have constructively been guaranteed
to hold, for a special $\lambda-$dependence of
model~(\ref{NyTSpt}), at
any integer $N\geq 2$ (see more details
in section \ref{uvodnik} below).

In what follows these observations
will inspire and enable us to simulate,
at any $N$, the
QC history starting at a conventional Hermitian
$N-$level oscillator Hamiltonian at $\tau=0$.
In the opposite, ``asymptotic'' extreme with $\tau \to
1$ the system will be found to collapse
into a complete (i.e.,
$N-$tuple) energy-level degeneracy.

\section{QC singularity \label{uvodnik}}

We will see below that
during the whole evolution process our model is such that
its nontrivial Hilbert-space
metric $\Theta^{(N)}(\tau)$ can be calculated in closed form.
This will enable us to attribute
the phenomenon of the QC collapse to
a manifest
interplay between the {\it ad hoc\,}
time-dependence of the Hamiltonian
and the equally {\it ad hoc\,}
time-dependence of the
Hilbert-space physical inner-product metric.

\subsection{Parametrization}

One of the simplest forms of the above-mentioned
$\lambda-$parametrizations of the couplings in Hamiltonians
(\ref{NyTSpt}) is given by the formula which was proposed in paper
\cite{tridiagonal},
 \be
 g_{n}^{(PT)}({\lambda})=\sqrt{n(N-n)(1-{\lambda})} 
 \,,\ \ \ \
 \ \ \ \
 n = 1, 2, \ldots, N-1\,.
 \label{paramo}
 \ee
The merit of this parametrization is that in the whole interval of
$\lambda \in (0,1)$ (or formally even of  $\lambda \in (0,\infty)$)
it guarantees the reality of the energy spectrum. The second merit of the
$\lambda-$parametrization~(\ref{paramo}) lies in the fact that the
boundary value of $\lambda=0$ strictly separates the stable dynamical
quantum regime (with $\lambda > 0$ yielding the real, ``observable''
$N-$plet of bound state energies) from the
half-axis of $\lambda<0$
(for which the system ceases to be observable).

In the algebraic terminology of paper \cite{maximal}
and of the older literature \cite{Kato,Heiss,Heissb} one encounters the so
called Kato's exceptional point (EP) of the $N-$th order at $\lambda=0$.
In paper \cite{tridiagonal} we restricted our attention to the
models with small $\lambda$s, therefore. We found that
Eq.~(\ref{paramo}) may be further reparametrized in terms of another
time variable $t$ which
has been found to measure a recovery from the QC singularity
in a broader physical multiparametric domain $ {\cal D}$,
 \be
 \lambda \ \to \ \lambda_n(t) = t+t^2+\ldots+t^{J-1}+G_n t^J\,\ \ \ \ \
  n = 1, 2, \ldots, N\,,\ \ \ \ J =[N/2]
  \,.
 \label{lobkov}
 \ee
The alternative time-parameter $t \in (0, \infty)$
appeared suitable for our having
Hamiltonian $ H^{(N)}_{(PT)}$ more
comfortably tractable
near its EP {\it alias\,} QC singularity.

\subsection{Time-dependent metric}

At small $t$s the evolution can be interpreted as the motion of
system away from QC, towards a stable and less anisotropic
dynamical regime.
Parametrization (\ref{lobkov}) proves useful, first of all, at the
very short times $t \ll 1$ at which it effectively rescales and
magnifies the interior of ${\cal D}$ in the vicinity of EP.
Simultaneously, such an {\em ad hoc} change of scale did not lower
the number of degrees of freedom -- one could still work with as
many as $J$ alternative coupling constants $G_n\geq 0$.

Trivial selection of special couplings $G_n=0$
realizes a one-parametric, simplified but still instructive
reduction of
the picture of dynamics. After the hypothetical start of evolution at
the $\lambda=0$ singularity one ultimately
reaches a manifestly Hermitian
regime at $\lambda=1$. In this sense, parameter
$\lambda=\lambda(t)\in (0,1)$
measures the recovery.

At any fixed value of $N$ and for any suitable Hamiltonian the physics
varies with inner product (\ref{3}).
The reconstruction of all of the eligible metrics $\Theta$ may be
found summarized in our dedicated work~\cite{SIGMAdva}. It
can be summarized as
starting from
Schr\"{o}dinger equation
 \be
  \left [H^{(N)}_{(PT)} \right ]^\dagger |\psi_n^{(N)} \kkt =
   E_n^{(N)} \, |\psi_n^{(N)} \kkt\,,
  \ \ \ \ \ \ n = 0, 1, \ldots, N-1\,
  \label{SEcon}
  \ee
where we replaced
Hamiltonian $H^{(N)}_{(PT)}{}$ by its Hermitian conjugate
in ${\cal H}^{(F)}$.
The complete solution of the new equation opens the way towards the
reconstruction of any metric from its spectral representation
 \be
 \Theta^{(N)}_{(\vec{\kappa})}(t)
 =\sum_{n=1}^N\,|\,\psi_n^{(N)}(t)\kkt\, \kappa_n \,\bbr \psi_n^{(N)}(t)
 |\,.
 \label{alice}
 \ee
All of the parameters $\kappa_n>0$ are freely variable. This
is an ambiguity which reflects the absence
of an exhaustive information about the physics behind the
quantum
system in question (cf., e.g.,
Ref.~\cite{Geyer} for explanation).

In Refs.~\cite{SIGMAdva} and \cite{lotor} we
discussed the general recipe
from a more formal point of view
by which multiple
suitable metrics $\Theta$ may always be assigned to a given
Hamiltonian $H$ via Eq.~(\ref{alice}). We emphasized there that
the construction is always ambiguous.
Now, it is worth adding that the sufficiency of the solution of the
auxiliary conjugate Schr\"{o}dinger equation (\ref{SEcon})
has to be perceived as a another, serendipitious  merit
of the models living in finite dimensions $N < \infty$.

\subsection{Two by two example\label{333}}

Let us pick up $N=2$ and study formula
(\ref{alice}) in more detail. Firstly, let us
change the variables, $t \to r=r(t)=\sqrt{t}>0$, yielding
 \be
 \left [H^{(N)}_{(PT)} \right ]^\dagger
 =\left[ \begin {array}{cc} -1&-\sqrt
{1-{r}^{2}}\\{}\sqrt {1-{r}^{2}}&1\end {array} \right]\,.
 \label{hadve}
 \ee
In terms of the pair of abbreviations $u= \sqrt{1-r}$ and $ v =
\sqrt{1+r}$ we may then calculate the maximal and
minimal eigenvalues $E_+=r$ and $E_-=-r$ of (\ref{hadve}) as well as
the related respective real eigenvectors
 \be
 |\psi_+\kkt = [ \sqrt{1-r},-\sqrt{1+r} ]^T=[u,-v]^T\,,
 \label{eigvcp}
 \ee
 \be
 |\psi_-\kkt = [ \sqrt{1+r},-\sqrt{1-r} ]^T=[v,-u]^T\,
 \label{eigvcm}
 \ee
where the superscript $^T$ denotes transposition.

It remains for us to insert vectors (\ref{eigvcp}) and
(\ref{eigvcm}) in the spectral expansion of the metric. Once we fix
an inessential overall factor and denote $\kappa_+=\sin \alpha$ and
$\kappa_-=\cos \alpha$ with $0 <\alpha < \pi/2$ we get the general
$N=2$ metric-operator matrix 
 \be
 \Theta=\Theta^{(2)}_{[\alpha]}(r^2)=\left[ \begin {array}{cc}
 1+r\cos 2\alpha & -\sqrt{1-r^2}
 \\
{}
 -\sqrt{1-r^2}&1-r\cos 2\alpha
 \end {array} \right]\,.
 \label{inn2}
 \ee
Its elementary form facilitates the direct determination of its
eigenvalues,
 \be
 \theta_{\pm}=1\pm \sqrt{1-r^2 \sin^2 2\alpha}\,.
 \label{difce}
 \ee
One easily verifies that the requirement of the necessary positivity
of these eigenvalues is trivially satisfied at any square-root-time
$r=\sqrt{t}$ such that $0 < r < 1$.

We may conclude that the standard probabilistic interpretation of
our time-dependent $N=2$ THS QC quantum model is determined not only
by its one-parametric Hamiltonian (\ref{hadve}) but also by the
specification of the concrete value of variable $\alpha$. Via
Eq.~(\ref{inn2}) this choice selects one of the eligible inner
products (\ref{3}). This makes the Hilbert space of states fully
defined and unique, with an asymmetry {\it alias\,}
anisotropy of its geometry
measurable simply by the difference $1-r^2 \sin^2 2\alpha$.


\section{Anisotropy\label{zavodnik}}

Any energy-representing input Hamiltonian $H(\tau)$
admits many alternative,
nonequivalent predictions
of the results of measurements \cite{Geyer}.
Formally, this is caused by the existence of
multiple Hamiltonian-independent
parameters as sampled by $\alpha$ in Eq.~(\ref{inn2})
at $N=2$. The ambiguity must
be removed via additional, physics-based constraints.

The situation becomes slightly different
near the QC degeneracy because the EP-related
confluence of the
energy levels already implies
that a consistent picture of reality
and, in particular,
the necessary regular
inner-product metric ceases to be positive definite.

A universal remedy does not exist because
a small increase/decrease of time may cause a large
change of
the metric in general.
In what follows we will exclude
such a formally admissible non-perturbative
behavior of $\Theta(\tau)$ as unphysical.

\subsection{Metrics as functions of time $\tau$}

In our older paper \cite{I}
we constructed the well-behaved, extrapolation-friendly
metrics up to $N=3$ or, in an implicit form
(\ref{alice}), up to $N=5$. The choice
of parameters $\vec{\kappa}$ was
dictated by a half-intuitive requirement of
simplicity.
Vague as such a recipe might have been,
it found an independent support in a similarity of formulae
at several Hilbert-space dimensions.
The success of such a
choice of parameters
motivated also our present study.

Our present strategy will
be based on a change of philosophy.
In place of studying just
a vicinity of QC using a small EP-unfolding time $t$
({\it alias\,} $\tau\lessapprox 1$)
we will
search for an explicit description of
the QC-emergence process
in the whole interval of $\tau \in (0,1)$.
Thus, the time
variable ${\tau}$ will replace the not too
suitable time-like
parameter $\lambda(t)$ of Eq.~(\ref{lobkov})).
Running in opposite direction, i.e., from the
$\tau=0$ instant (at which our Hamiltonian is diagonal)
to the QC limit of $\tau \to 1$
(in which our Hamiltonian becomes
non-diagonalizable and merely Jordan-block
representable).


The choice of the new ``time of
collapse'' variable ${\tau}=\sqrt{1-{\lambda}}$
will enable us to treat the
fall of our system into its singularity as a process which started
long before the catastrophe.
In the final QC limit $\tau \to 1$ the
time-dependent vectors $|\psi_n\kkt$
of Eq.~(\ref{SEcon})
can be then perceived as getting mutually parallel.
The operator (\ref{alice}) itself degenerates to a
singular but
particularly elementary matrix of rank one of course.

A disadvantage
of the replacement of $t$ or $\lambda(t)$
by $\tau$ might be
that
the construction of the metric appeared easier at small
$\lambda$.
Fortunately, near the
opposite extreme of $\tau=0$ the transition $t \to \tau$ simplifies
the Hamiltonians themselves,
 \ben
 H^{(2)}({\tau}) = \left [\begin {array}{cc}
 -1&{\tau}\\{}-{\tau}&1\end {array}\right ]
 \,,\ \ \ \
 H^{(3)}({\tau}) = \left [\begin {array}{ccc}
  -2&\sqrt{2}\,{\tau}  &0  \\
 -\sqrt{2}\,{\tau}&0    &\sqrt{2}\,{\tau}\\
  0&-\sqrt{2}\,{\tau}&2
 \end {array}\right ]
 \,,\ \ \ \
 \een
 \be
 \ \ \ \ \
 H^{(4)}({\tau}) =
 \left [\begin {array}{cccc}
  -3&\sqrt{3}\,{\tau}   &0  &0\\
 -\sqrt{3}\,{\tau}&-1   &2\,{\tau}  &0\\
  0&-2\,{\tau}  &1 &\sqrt{3}\,{\tau}\\
  0&0&-\sqrt{3}\,{\tau}&3
 \end {array}\right ]\,,\ \ldots\,.
 \label{radaham}
 \ee
One is led to a replacement of
expansion (\ref{alice}) by a less subtle, brute-force
technique, rendered possible by the tridiagonal-matrix
form of Hamiltonian.

Such an idea proved
productive, indeed.

\begin{thm}
The metrics $\Theta^{(N)}({\tau})$ compatible with
Hamiltonians (\ref{radaham}) may be sought in the finite-sum form
 \be
 \Theta^{(N)}({\tau})=
 \sum_{j=1}^{N}\,(-{\tau})^{j-1} {\cal M}^{(N)}(j)\,
 \label{hlavni}
 \ee
with sparse-matrix coefficients
 \be
 {\cal M}^{(N)}(1)= \left[
 \begin {array}{cccc}
  \alpha_{11}(1)&0&\ldots&0
 \\
 {}
 0&\alpha_{12}(1)&\ddots&\vdots
 \\
 {}\vdots&\ddots&\ddots&0
 \\
 {}0&\ldots&0&\alpha_{1N}(1)
 \end {array} \right]\,,
 \label{hlavni1}
 \ee
 \be
 {\cal M}^{(N)}(2)= \left[
 \begin {array}{cccccc}
  0&\alpha_{11}(2)&0&\ldots&\ldots&0
 \\
 {}
 \alpha_{21}(2)&0&\alpha_{12}(2)&0&\ldots&0
 \\{}
 0&\alpha_{22}(2)&0&\alpha_{13}(2)&\ddots&\vdots
 \\{}
 \vdots&\ddots&\ddots&\ddots&\ddots&0
 \\{}
 0&\ldots&0&\alpha_{2,N-2}(2)&0&\alpha_{1,N-1}(2)
 \\
 0&\ldots&{}\ldots&0&\alpha_{2,N-1}(2)&0
 \end {array} \right]\,,
 \ \ \ \ \
 \label{hlavni2}
  \ee
 \be
  \ \ \ \ \
 {\cal M}^{(N)}(3)= \left[
 \begin {array}{ccccccc}
  0&0&\alpha_{11}(3)&0&\ldots&\ldots&0
 \\
 {}
 0&\alpha_{21}(3)&0&\alpha_{12}(3)&0&\ldots&0
 \\{}
 \alpha_{31}(3)&0&\alpha_{22}(3)&0&\alpha_{13}(3)&\ddots&\vdots
 \\{}
 0&\alpha_{32}(3)&\ddots&\ddots&\ddots&\ddots&0
 \\{}
 0&\ddots&\ddots&0&\alpha_{2,N-3}(3)&0&\alpha_{1,N-2}(3)
 \\{}
 \vdots&\ldots&0&\alpha_{3,N-3}(3)&0&\alpha_{2,N-2}(3)&0
 \\
 0&\ldots&{}\ldots&0&\alpha_{3,N-2}(3)&0&0
 \end {array} \right]\,,
 \label{hlavni3}
   \ee
etc.

\end{thm}

\begin{proof}
follows  from the observation that
at any $k=1,2,\ldots,N$,
the set of all of the non-vanishing elements of
matrix ${\cal M}^{(N)}(k)$ may be compressed and
arranged into an auxiliary, $k$ by
$(N-k+1)-$dimensional array
 \be
  \alpha{(k)}= \left[ \begin {array}{ccccc}
   \alpha_{11}(k)&\alpha_{12}(k)&\alpha_{13}(k)&\ldots&\alpha_{1,N-k+1}(k)
   \\
   \alpha_{21}(k)&\alpha_{22}(k)&\alpha_{23}(k)&\ldots&\alpha_{2,N-k+1}(k)
 \\
   \vdots& \vdots& \vdots& & \vdots
   \\
   \alpha_{k1}(k)&\alpha_{k2}(k)&\alpha_{k3}(k)&\ldots&\alpha_{k,N-k+1}(k)
 \end {array} \right]\,.
 \label{referendu}
 \ee
with $\alpha_{11}(k)={\cal
M}^{(N)}_{1k}(k)$, etc. This reduces the proof to the
inspection of the set of $N^2$ Dieudonn\'es
linear algebraic compatibility relations written in the matrix form
 \be
 H^\dagger \Theta=\Theta\,H
 \label{dieudd}
 \ee
not all of which are independent (cf. Ref. \cite{SIGMAdva} for
details).
\end{proof}

\subsection{Onset of the process of degeneracy}

From the point of view of potential applications of formula
(\ref{hlavni}) it is important that at the beginning of the fall
into instability (i.e., at the far-from-QC instant ${\tau}=0$) the
Hamiltonians (\ref{radaham}) will all coincide with the respective
truncated and diagonal (i.e., Hermitian) harmonic-oscillator-like
matrices. In this picture the spectrum of energies
$E_n^{(N)}({\tau})$ will remain real but shrinking with the growth
of the innovated time ${\tau}$. In the limit ${\tau} \to 1$, i.e.,
at the very end of the fall of the system into QC singularity, the
spectrum becomes completely degenerate, $E_n^{(N)}(1)=0$,
$n=0,1,,\ldots,N-1$.

In the latter limit the Hamiltonian
(i.e., matrix $H^{(N)}(1)$) ceases to be
diagonalizable and loses its standard physical tractability and
interpretation. {\em Vice versa}, from the point of view of physics
the description of the evolution as generated by $H^{(N)}(\tau)$
will change at $\tau=1$, requiring an introduction of some new
degrees of freedom beyond this instant, i.e., at $\tau>1$. The study
of such a discontinuous
switch to a new form of Hamiltonian at later times lies
already beyond the scope of present paper:
Interested readers might consult, e.g., a dedicated study \cite{passage}.
Beyond the framework of quantum theory,
examples of such an EP-mediated phase transition
may be found, e.g., in magnetohydrodynamics~\cite{suwem}.

Temporarily, let us now return, in the context of interpretations,
to the times of recovery $t$ or $\lambda(t)$. In these variables,
the motion beyond the end-of-the-interval $\lambda(t)=1$ appears
much less exotic. Typically, the energies would not feel the change
at all (cf., e.g., Eqs.~(\ref{E5}), (\ref{E6}) or (\ref{E7})).
Still, the values
of $\lambda =\lambda^{(outer)}>1$ remain {\em mathematically}
less interesting because in the ``outer''  interval
of parameters our
Hamiltonian matrices become
Hermitian.
Thus,
it is natural to require that in
the extrapolated dynamical regime the
metric remains constant and trivial, $\Theta^{(outer)} \equiv
I$. In other words, we propose to match the
non-Hermitian and Hermitian dynamical regimes strictly at
$\lambda(t)=1$, with
 \be
  \lim_{\tau \to 0}\ \Theta^{(N)}({\tau}) = I \,.
  \label{restr}
  \ee
One of the consequences of such a discontinuation of the model
may be seen in the subsequent most natural reinterpretation
of the point of matching $\lambda(t)=1$: The process of the
QC degeneracy
is expected to proceed only at $\lambda(t)<1$.
The optimal
QC-related metrics should be then required  continuous
{\em just at the relevant times\,} $\tau \in (0,1)$, i.e.,
in particular,
up to
the very instant
of the EP degeneracy.

\subsection{Example: $N=2$}

A deeper phenomenological meaning of  requirement (\ref{restr})
may be most immediately illustrated via the
$N=2$ model. Recalling the set of {\em all\,} metrics
(\ref{inn2}) we notice that they are
numbered by the single optional real
variable $\alpha$.
It is easy to see that each deviation of this
parameter from its unique,
anisotropy-minimizing and
extrapolation-friendly
value as deduced from formula
(\ref{difce}) would necessarily violate
constraint (\ref{restr}).

The existence of the closed formula for the spectrum
enables us to see that the
influence of $\alpha$ changes from very weak (in the QC vicinity,
i.e., at $t \ll 1$) to very strong (near the Hermitian dynamical
regime where $\lambda(t)=1$). The extrapolation-friendly
choice of $\alpha = \pi/4$
of  papers \cite{I} or \cite{lotor} appears exceptional.
{\em Solely} for this choice of the
free parameter the difference between the eigenvalues of the metric
will strictly vanish at $\lambda(t)=1$. We may just repeat that for
the value of $\alpha = \pi/4$ the $\tau=0$ instant really carries
the meaning of a {Hermitian} onset of the fall into QC singularity.

Once we preserve the latter exceptional choice of the parameter
$\alpha = \pi/4$ at all times $\tau \in (0,1)$, the difference between
the two eigenvalues of the $N=2$ metric (measuring a Hilbert-space
anisotropy) remains always minimized (cf.
Eq.~(\ref{difce})). In this formulation, the selection of a
``minimal Hilbert-space anisotropy'' principle
of paper \cite{lotor} is certainly confirmed as
optimal.



\section{Eigenvalues of the metrics\label{prevodnik}}

In the spirit of the methodical
project as outlined in paper \cite{lotor} the uniqueness of the
choice of $\alpha = \pi/4$ at $N=2$ should be, {\em mutatis mutandis},
extended and amended
to apply at any $N$. Preliminarily,
such an idea has been tested and found feasible
in Ref.~\cite{I} where we recalled
formula~(\ref{alice})
and where we managed to
evaluate,
up to $N=5$, the
ketket-eigenvectors $|\psi_n^{(N)} \kkt\, $
in closed form.

Now we intend to amend the
recipe and to find and formulate a more general result.
Our task may be separated into two subtasks.
In the first one
(to be dealt with in this section)
the problem
will be reconsidered
at a few
smallest dimensions $N$.
We will reveal that
a new and promising
guide to extrapolations in $N$
can and should be sought in a
certain very regular sparse-matrix pattern
emerging in the
formulae for the metrics $\Theta^{(N)}(\tau)$
[cf. also
Eq. Nr. (10) in \cite{I} in this respect].

Secondly,
in a genuine climax of our present constructive efforts
(and in a way described in subsequent section \ref{svod})
we will find
that the latter observation opens the way towards
a remarkably efficient study and closed-form evaluation
of the eigenvalues
$\theta^{(N)}_n(\tau)$
of the metrics.
Very well reflecting
both the anisotropy and asymptotic degeneracy of the
physical Hilbert space
and, hence, carrying a perceivably more useful
information about dynamics
than the
matrix elements of the metric-operator
$N$ by $N$ matrices $\Theta^{(N)}(\tau)$ themselves.

\subsection{$N=2$, revisited \label{revi2}}

The explicit construction of metric
$\Theta^{(N)}({\tau})$
via auxiliary Schr\"{o}dinger Eq.~(\ref{SEcon}) is
not too easy even at $N=2$, i.e., for our first nontrivial
QC-related Hamiltonian matrix
 \be
 H^{(2)}({\tau})=\left[ \begin {array}{cc}
  -1&{\tau}\\-{\tau}&1\end {array}
 \right]\,.
 \label{neher2}
 \ee
The efficiency of this construction remains comparable with
the brute-force solution of Eq.~(\ref{dieudd})
(cf.
subsection \ref{33} above). Nevertheless,
it still makes sense to
rederive metric
$\Theta^{(2)}$, by the amended method,
for the pedagogical purposes.

We may start from the real-matrix ansatz
 \be
 \Theta^{(2)}_{(\vec{\kappa})}({\tau})=\left[ \begin {array}{cc}
  a&b\\b&d\end {array}
 \right]\,
 \label{neher2}
 \ee
with the subscripted vector $\vec{\kappa}$
containing two arbitrary positive components.
Next we fix an overall multiplication constant by
setting the determinant equal to one. This enables us to put
$b=\sinh \nu$ and choose $\varepsilon =\pm 1$ in $a =
\varepsilon\,\cosh \nu\,\exp \varrho$ and $d = \varepsilon\,\cosh
\nu\,\exp (-\varrho)$.

Both of the new parameters $\nu$ and $\varrho$
are assumed real. The metric must be positive so that we may
only use $\varepsilon = 1$. Finally we check that the matrix
constraint (\ref{dieudd}) degenerates to the single,
time-reparametrization item
 \be
  {\tau}= -\frac{\tanh \nu}{\cosh \varrho}\,.
 \label{bla}
 \ee
Our conclusion is that for any given ${\tau} \in (0,1)$ we may
choose {\em any} real $\varrho \in (0,\varrho_{max})$ (note that
this is the parameter which makes the main diagonal of the metric
asymmetric).

This choice enables us to evaluate
$\nu=\nu({\tau},\varrho)$ from the latter equation (this implies
that at a fixed time, the value of $\varrho_{max}$ must be such that
$\cosh \varrho_{max} =1/{\tau}$). Summarizing, we may set
$\alpha_{11}(1)=\cosh \nu \exp \varrho$, $\alpha_{12}(1)=\cosh \nu
\exp (- \varrho)$ and $\alpha_{11}(2)=\sinh \nu$ in
Eq.~(\ref{hlavni}) at $N=2$.
The resulting eigenvalues of the metric
 \be
 \theta_\pm = \cosh \nu \cosh \varrho
  \pm \sqrt{\cosh^2 \nu \cosh^2 \varrho -1}\,
  \label{onen}
  \ee
are both, by construction, positive.

At the very start of the fall of the system into the catastrophe,
i.e., at ${\tau}=0$ one has $ \varrho_{max}(0) =\infty$ so that
there is no upper bound imposed upon $\varrho(0)$. Still, as long as
one might like to have the trivial, isotropic initial value of
$\Theta^{(2)}(0)\sim I$ (implying the special choice of $\nu(0)=0$
and $\varrho(0)=0$) the resulting metric becomes, up to the
above-mentioned irrelevant overall multiplication factor, unique at
$\tau=0$.

During the subsequent growth of $\tau <1$ the requirement of the
minimization of the anisotropy leads to the rule $\varrho(\tau)=0$
(cf. Eq.~(\ref{onen})) so that the remaining variable $\nu<0$ may
now be interpreted as another version of the time of the QC
degeneracy which is just rescaled and, incidentally,
inverted (cf. Eq.~(\ref{bla})).

Once we return to the standard variables we get our unique and
minimally anisotropic metric in the virtually trivial form
 \ben
 \Theta^{(2)}=\left[ \begin {array}{cc}
 1 & -{\tau}
 \\
{}
 -{\tau}&1
 \end {array} \right]=I-{\tau}\,J\,.
 \een
From this formula we may deduce the special, minimally anisotropic
version of eigenvalues in the form compatible with their more strongly
anisotropic generalization (\ref{difce}).

\subsection{$N=3$}

Whenever one tries to move to the higher matrix dimensions $N$ one
encounters the technical problem of an increase of multitude of
parameters. In the first nontrivial case with $N=3$
let us first follow the $N=2$ guidance (cf. the ultimate choice of
$\varrho = 0$ in the preceding paragraph \ref{revi2}) and let us
omit the discussion of the metrics with an asymmetric form of their
main diagonal.

Once we also keep ignoring the other, irrelevant though still
existing overall factor, we are, after some straightforward
manipulations using Eq.~(\ref{dieudd}), left with the last free
parameter $g$ in the metric
 \be
  \Theta^{(3)}({\tau})= \left[ \begin {array}{ccc} 1&-\sqrt {2}g{\tau}&g{{\tau}}^{2}
 \\\noalign{\medskip}-\sqrt {2}g{\tau}&2\,g-1+g{{\tau}}^{2}&-\sqrt {2}g{\tau}
 \\\noalign{\medskip}g{{\tau}}^{2}&-\sqrt {2}g{\tau}&1\end {array} \right]\,.
 \label{refere}
 \ee
Among its three readily obtainable eigenvalues
 \be
 \theta_1=g{{\tau}}^{2}+g-\sqrt {4\,{g}^{2}{{\tau}}^{2}+{g}^{2}-2\,g+1}
 \,,
 \ \ \ \ \
 \theta_2=1-g{{\tau}}^{2}\,,
 \ \ \ \ \
 \theta_3=g{{\tau}}^{2}+g+\sqrt {4\,{g}^{2}{{\tau}}^{2}+{g}^{2}-2\,g+1}
 \label{difere}
 \ee
the middle one (with an inverted-parabola
dependence on ${\tau}$) remains
positive for the parameters
  $g<1/{\tau}^2$.

The change of sign of the remaining two
eigenvalues takes place at
the curves $g=1/{\tau}^2$ and $g=1/(2-{\tau}^2)$ in the $g-{\tau}$
plane. As a consequence,
the correct and unique choice of the parameter is $g=1$,
yielding again the unique metric
 \be
 \Theta^{(3)}=I  - {\tau}\,\left[ \begin {array}{ccc}
 0 & \sqrt{2}&0
 \\
{}
 \sqrt{2}&0&\sqrt{2}
 \\
{}
 0&\sqrt{2}&0
 \end {array} \right]+ {\tau}^2 J\,
 \label{formu3}
  \ee
with the expected $\tau-$dependence of the eigenvalues
as given by Eq.~(\ref{difere}).

\subsection{$N=4$}

In the next step of our constructive considerations we are getting
beyond the formulae derived in older papers.
We succeeded because already the
$N=3$ formula (\ref{formu3}) offered the hint.
Thus, making use of the analogy and performing an
extrapolation, it just proved sufficient to verify that the
following tentative candidate for the metric
 \be
 \Theta^{(4)}=\left[ \begin {array}{cccc} 1&-\sqrt {3}{\tau}&\sqrt
 {3}{{\tau}}^{2}&-{{\tau}}^{3}\\{}-\sqrt
 {3}{\tau}&1+2\,{{\tau}}^{2}&-2\,{\tau}-{{\tau}}^{3}&\sqrt {3}{{\tau}
 }^{2}\\{}\sqrt
 {3}{{\tau}}^{2}&-2\,{\tau}-{{\tau}}^{3}&1+2\,{{\tau}}^{2}&- \sqrt
 {3}{\tau}\\{}-{{\tau}}^{3}&\sqrt {3}{{\tau}}^{2}&-\sqrt {3}{\tau}&1
 \end {array} \right]
 \label{refere4}
 \ee
obeys all the necessary and sufficient requirements.
They include
the validity of the Dieudonn\'e's Eq.~(\ref{dieudd}) as well as the
feasibility of evaluation of the
$\tau-$dependent eigenvalues of the candidate for the metric.
We immediately see that they behave as they should,
 \be
  \left\{\theta_1,\ldots,\theta_4
  \right\}=
 \left\{ 1-3\,{\tau}+3\,{{\tau}}^{2}-{{\tau}}^{3},
  1-{\tau}-{{\tau}}^{2}+{{\tau}}^{3},1+3\,{\tau}+3\,{{\tau}}^{
 2}+{{\tau}}^{3},1+{\tau}-{{\tau}}^{2}-{{\tau}}^{3} \right\}.
 \label{difere4}
 \ee
Their
correct QC behaviour at $\tau = 1$ really deserves an explicit
graphical display as provided by Figure~\ref{bigthee2}.

%

\begin{figure}[t]                     
\begin{center}                         
\epsfig{file=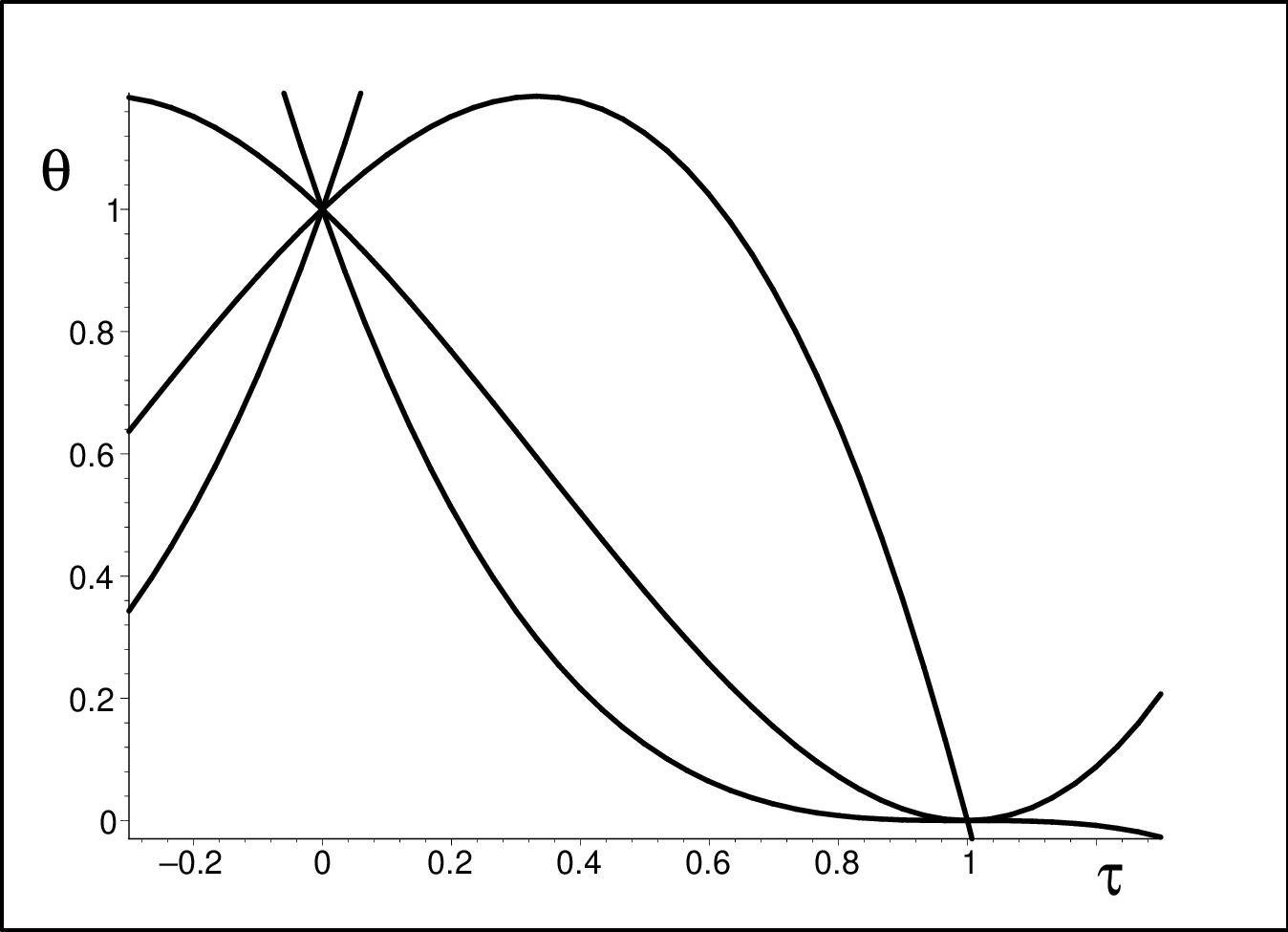,angle=270,width=0.5\textwidth}
\end{center}                         
\vspace{-2mm}\caption{Eigenvalues (\ref{difere4}) of the 
$N=4$ metric (\ref{refere4})
as functions of time 
$\tau$.}
 \label{bigthee2}
\end{figure}
%

\section{Extrapolations \label{svod}}

\subsection{Metrics between $N=5$ and
$N=7$}


We may now combine the results of preceding section with the
contents of Theorem 1. Using an elementary insertion in
Eq.~(\ref{dieudd}) we may easily prove that in the expansions
(\ref{hlavni}) of the metrics with minimal anisotropy the
diagonal-matrix coefficients (\ref{hlavni1}) may be defined, at all
$N$, by the elementary formula
 $$
 \alpha_{1n}(1)=1\,,\ \ \ \ n = 1, 2,
 \ldots, N\,.
 $$
%
%
%
%
Similarly, the closed formula is also available for the antidiagonal
coefficients in ${\cal M}^{(N)}(N)$,
 $$
 \alpha_{n1}(N)=1\,,\ \ \ \ n = 1, 2,
 \ldots, N\,.
 $$
Next, the bidiagonal matrix coefficients (\ref{hlavni2}) may be
defined, at all $N$, by the slightly less elementary general formula
 $$
 \alpha_{1n}(2)=\alpha_{2n}(2) = \sqrt{n(N-n)}\,,\ \ \ \ n = 1, 2,
 \ldots, N-1\,.
 $$
Due to the easily verified symmetry, the analogous formula exists
for the coefficients in ${\cal M}^{(N)}(N-1)$,
 $$
 \alpha_{n1}(N-1)=\alpha_{n2}(N-1) = \sqrt{n(N-n)}\,,\ \ \ \ n = 1, 2,
 \ldots, N-1\,.
 $$
Up to now, unfortunately, we did not succeed in an extension of
these observations to the tridiagonal sparse matrix coefficients
(\ref{hlavni3}), etc. Nevertheless, we believe that the task is not
impossible. This belief seems supported by Theorem 1, i.e., by the
reducibility of the $N$ by $N$ matrices ${\cal M}^{(N)}(k)$ with
$k=3,4,\ldots$ to the respective auxiliary $ k$ by $N-k+1$ arrays
containing the non-vanishing matrix elements $\alpha_{jm}(k)$ of
${\cal M}^{(N)}(k)$.

The first missing set of coefficients occurs at $N=5$. Its values
 \ben
 \alpha_{11}(3)=\alpha_{13}(3) = \alpha_{31}(3)=\alpha_{33}(3) =
 \sqrt{6}\,,
 \een
 \ben
  \alpha_{12}(3)=\alpha_{21}(3) =
 \alpha_{23}(3)=\alpha_{32}(3) = 3\,,\ \ \ \ \alpha_{22}(3)=4\,.
 \een
should be better rewritten in the compact form of array
 \be
  \alpha{(3)}= \left[ \begin {array}{ccc}
   \sqrt{6}&3&\sqrt{6}
 \\
 3&4&3
 \\
 \sqrt{6}&3&\sqrt{6}
 \end {array} \right]\,,\ \ \ \ \
 N=5\,.
 \label{referenc}
 \ee
It makes sense to complemented this result by the next, $N=6$
formula
 \ben
 \alpha_{11}(3)=\alpha_{14}(3) = \alpha_{31}(3)=\alpha_{34}(3) =
 \sqrt{10}\,,\ \ \ \alpha_{21}(3)= \alpha_{24}(3)=4\,,
  \een
   \ben
 \alpha_{12}(3)=\alpha_{13}(3) = \alpha_{32}(3)=\alpha_{33}(3) =
 3\,\sqrt{2}\,,\ \ \ \ \
 \alpha_{22}(3)=
 \alpha_{23}(3)=6 \,
 \een
which we derived using the brute force construction based on
Eq.~(\ref{dieudd}). It again deserves the compact presentation as
array
 \be
  \alpha{(3)}= \left[ \begin {array}{cccc}
   \sqrt{10}&3\sqrt{2}&3\sqrt{2}&\sqrt{10}
 \\
 4&6&6&4
 \\
   \sqrt{10}&3\sqrt{2}&3\sqrt{2}&\sqrt{10}
 \end {array} \right]\,,\ \ \ \ \
 N=6\,.
 \label{referendu}
 \ee
The closed form of the latter result indicates that there might
exist a not too complicated extrapolation recipe, with the help of
which we would be able to determine the unique, minimally
anisotropic metric at any dimension $N$. This belief seems further
supported by the regularity and apparent extrapolation-friendliness
of the next two sparse-matrix ``missing'' coefficients
 $$
 {\cal M}^{(7)}(3)=\left[ \begin {array}{ccccccc} 0&0&\sqrt {15}&0&0&0&0
 \\0&5&0&
 \sqrt {30}&0&0&0
 \\\sqrt {15}&0&8&0&6&0&0\\0&
 \sqrt {30}&0&9&0& \sqrt {30}&0
 \\0&0&6&0&8&0&\sqrt {15}\\0&0&0&
 \sqrt {30}&0&5&0\\0&0&0&0&\sqrt {15 }&0&0\end
{array} \right]
 $$
and 
 $$
 {\cal M}^{(7)}(4)=
 \left[ \begin {array}{ccccccc} 0&0&0&2\,\sqrt {5}&0&0&0\\
 {}0&0&2\,\sqrt {10}&0&2\,\sqrt {10}&0
 &0\\{}0&2\,\sqrt {10}&0&6\,\sqrt {3}&0&2\,
 \sqrt {10}&0\\{}2\,\sqrt {5}&0&6\,
 \sqrt{3}&0& 6\,\sqrt {3}&0&2\,\sqrt {5}\\{}0&2\,
 \sqrt {10} &0&6\,\sqrt {3}&0&2\,\sqrt {10}&0\\{}0&0&2\,
  \sqrt {10}&0&2\,\sqrt {10}&0&0\\{}0&0&0 &2\,\sqrt {5}&0&0&0
  \end {array} \right]\,.
 $$
They were obtained, with the assistance of the computerized
symbolic manipulations, by the brute-force solution of the set of 49
linear algebraic Eqs.~(\ref{dieudd}).

\subsection{Eigenvalues at arbitrary $N$}

\begin{table}[th]
\caption{Pascal triangle for coefficients $C_{1n}^{(N)}$ in
Eq.~(\ref{nejdul})
 } \label{pexp3b}
\begin{center}
\begin{tabular}{||c|ccccccccccccccc||}
\hline \hline
  N &
  \multicolumn{15}{c||}{ }\\
 \hline \hline
 1&&&&&&&&1&&&&&&&\\
 2&&&&&&&1&&{1}&&&&&&\\
 3&&&&{}&&1&&2&&1&&&&&\\
 4&{}&&&&1&&3&&{3}&&{1}&&{}&&\\
 5&&&&1&&4&&6&&4&&1&&{}&\\
 6&&&1&&5&&10&&{10}&&{5}&&{1}&&{}\\
 7&&1&&6&&15&&20&&15&&6&&1&\\
 8&1&&7&&21&&35&&{35}&
 &{21}&&{7}&&{1}\\
 $\vdots$&&&& & & &
  & & $\ldots$& & &&& &\\
  \hline \hline
\end{tabular}
\end{center}
\end{table}


\begin{table}[th]
\caption{Pascal-like triangle for coefficients $C_{2n}^{(N)}$  in
Eq.~(\ref{nejdul})} \label{pexp3a}
\begin{center}
\begin{tabular}{||c|ccccccccccccccc||}
\hline \hline
  N &
  \multicolumn{15}{c||}{ }\\
 \hline \hline
 2&&&&&&&1&&{-1}&&&&&&\\
 3&&&&{}&&1&&0&&-1&&&&&\\
 4&{}&&&&1&&1&&{-1}&&{-1}&&{}&&\\
 5&&&&1&&2&&0&&-2&&-1&&{}&\\
 6&&&1&&3&&2&&{-2}&&{-3}&&{-1}&&{}\\
 7&&1&&4&&5&&0&&-5&&-4&&-1&\\
 8&1&&5&&9&&5&&{-5}&
 &{-9}&&{-5}&&{-1}\\
 $\vdots$&&&& & & &
  & & $\ldots$& & &&& &\\
  \hline \hline
\end{tabular}
\end{center}
\end{table}


\begin{table}[th]
\caption{Pascal-like triangle for coefficients $C_{3n}^{(N)}$  in
Eq.~(\ref{nejdul})} \label{pexp3c}
\begin{center}
\begin{tabular}{||c|ccccccccccccccc||}
\hline \hline
  N &
  \multicolumn{15}{c||}{ }\\
 \hline \hline
  3&&&&{}&&1&&-2&&1&&&&&\\
 4&{}&&&&1&&-1&&{-1}&&{1}&&{}&&\\
 5&&&&1&&0&&-2&&0&&1&&{}&\\
 6&&&1&&1&&-2&&{-2}&&{1}&&{1}&&{}\\
 7&&1&&2&&-1&&-4&&-1&&2&&1&\\
 8&1&&3&&1&&-5&&{-5}&
 &{1}&&{3}&&{1}\\
 $\vdots$&&&& & & &
  & & $\ldots$& & &&& &\\
  \hline \hline
\end{tabular}
\end{center}
\end{table}


\begin{table}[th]
\caption{Pascal-like triangle for coefficients $C_{4n}^{(N)}$ in
Eq.~(\ref{nejdul}) } \label{pexp3d}
\begin{center}
\begin{tabular}{||c|ccccccccccccccc||}
\hline \hline
  N &
  \multicolumn{15}{c||}{ }\\
 \hline \hline
   4&{}&&&&1&&-3&&{3}&&{-1}&&{}&&\\
 5&&&&1&&-2&&0&&2&&-1&&{}&\\
 6&&&1&&-1&&-2&&{2}&&{1}&&{-1}&&{}\\
 7&&1&&0&&-3&&0&&3&&0&&-1&\\
 8&1&&1&&-3&&-3&&{3}&
 &{3}&&{-1}&&{-1}\\
 $\vdots$&&&& & & &
  & & $\ldots$& & &&& &\\
  \hline \hline
\end{tabular}
\end{center}
\end{table}

 \noindent
In the above-described constructions of the $N$ by $N$ matrices of
metric $\Theta^{(N)}$ we did not manage to find, unfortunately, any
obvious general extrapolation tendency or pattern. For this reason,
we turned attention from matrices to the perceivably
simpler-to-display $N-$plets of their eigenvalues
$\theta_n^{(N)}({\tau})$. At the first few values of $N$
we performed the brute-force calculations.
We met the ultimate success which may be given the following form
of proposition.

\begin{prop}
The time-dependent {eigenvalues} of $\Theta^{(N)}(\tau)$ may be written in
the form of polynomials
 \be
 \theta_n^{(N)}({\tau})
 =\sum_{k=1}^N\,C^{(N)}_{nk}\,{\tau}^{k-1}\,
 \label{nejdul}
 \ee
where the numerically evaluated values of the
coefficients $C^{(N)}_{nk}$ may be found listed,
up to $N=8$, in the
Pascal-like schemes of Tables \ref{pexp3b} - \ref{pexp3d}.
\end{prop}


 \noindent
After a cursory inspection of the latter four
Tables \ref{pexp3b} - \ref{pexp3d} one immediately
finds, in all of them, a regularities resembling
the well known Pascal triangle.
The analogy is almost perfect. It
enabled us to reveal the extrapolation pattern
and to specify the recurrences
for coefficients which appeared all solvable
in terms of binomial coefficients.
In other words, the
closed-form eigenvalues (\ref{nejdul}) of the metrics
have been obtained
for any time $\tau$ and for any matrix
dimension $N$. This is our present main result.

\begin{thm}

The time-dependent eigenvalues of metrics $\Theta^{(N)}(\tau)$
of Eq.~(\ref{hlavni}) are given by
formula
$$
\theta_k^{(N)}({\tau})=\sum_{m=1}^N\,C^{(N)}_{km}\,{\tau}^{m-1}\,,\
\ \ \ \ \ \ k=1,2,\ldots,N
$$
where $ C_{1n}^{(N)}=\left (\ba N-1\\n-1 \ea \right )\,$, $
C_{2n}^{(N)}=\left (\ba N-2\\n-1 \ea \right )-\left (\ba N-2\\n-2
\ea \right )\, $ and, in general,
  $$
 C_{kn}^{(N)}=\sum_{p=1}^k\,(-1)^{p-1}\, \left (\ba k-1\\p-1 \ea
 \right )\, \left (\ba N-k\\n-p \ea \right )\,,\ \ \ \
 k,n=1,2,\ldots,N\,.
 $$

\end{thm}

\begin{proof}
is straightforward and proceeds by mathematical induction.
\end{proof}


\section{Discussion\label{presummary}}


\subsection{Unbounded differential-operator Hamiltonians}

In the current literature the merits and applicability of the THS formalism
are
most often illustrated by the replacement of the
most common harmonic or anharmonic
oscillator by the Bender's and
Boettcher's \cite{BB} family of non-Hermitian
power-law-interaction
models
$H^{(BB)}(\delta)=p^2+{\rm i}^\delta\,x^{2+\delta}$
which are
characterized just by the
single real exponent $\delta\geq 0$.
In the historical
perspective such a choice was
surprising but
well motivated by the needs of development of
quantum field theory (cf., e.g., \cite{BM}) and/or of
perturbation expansion methods \cite{BGbf} -- \cite{Alvarezb}.

The use of nontrivial and, in general, manifestly
Hamiltonian-dependent Hilbert-space-metric operators may be
perceived as an important innovation of the model-building in
quantum theory. Nevertheless, whenever accepted as a sound
theoretical tool in physics, its mathematical consistency must
always be carefully reexamined. In this sense the
proof of nonexistence of the metric
operator for the most popular and phenomenologically highly relevant
$H^{(BB)}(\delta)$ \cite{Siegl}
makes the study of this particular Hamiltonian
far from being completed.
Hence, the illustration purposes are, in the eyes of mathematicians,
much better served
by the bounded-operator Hamiltonians
\cite{Geyer} as sampled in our paper.

This being said, it is still
possible to conclude that
in comparison with the
conventional quantum theory using selfadjoint operators
the
question of interpretation of the
Bender's and
Boettcher's unbounded-Hamiltonian
models, albeit still open \cite{Uwe},
remains inspiring.
Especially in the light of
the abstract mathematical
comments by Dieudonn\'{e} \cite{Dieudonne}
who pointed out
that not only the necessary
{\em ad
hoc} specification 
but even the very proof of existence
of the correct physical
Hilbert space of states
may be a highly nontrivial question.


\subsection{Quantum catastrophe as an EP-related concept}

The classical Thom's concept of a ``catastrophe'' \cite{III}
is based on the idea that in a certain dynamical regime
an infinitesimal change of one or more relevant parameters
leads to an abrupt change of the behavior of the system.
In Ref.~\cite{I} we
conjectured that
a quantum analogue
of such a singularity can be represented by
the Kato's \cite{Kato} EP singularity
at which a multiplet of quantum bound-state energies
merges and
ceases to be observable.

In spite of the fact that one of the older
reviews of the related theory
(viz., Ref.~\cite{Geyer}) already appeared
as early as in 1992,
an extension of its scope beyond the
domain of nuclear physics
was not too quick.
Fortunately, multiple extensions already do exist at present,
based on the
innovations of the theory
called ${\cal PCT}-$symmetric~\cite{Carl}
{\it alias\,}
pseudo- or
quasi-Hermitian \cite{ali}
or Krein-space self-adjoint~\cite{hlali}.
In this context the construction of illustrative quantum catastrophes
in \cite{I} 
has been
facilitated by the
availability
of Hamiltonians
for which at least some of
the Kato's EP singularities
(which occurred, traditionally, just at the
complex, ``unphysical'' values of parameters)
became experimentally accessible, 
in principle at least \cite{Makris,Makrisb}.

In our present paper we
decided to reanalyze, therefore,
some of the latter models,
with the emphasis put upon
the open questions
concerning
a smooth transition from
the classical to
quantum dynamics.
Our approach has been
based on the use of the
THS formalism, with the aim of a further
amendment of our understanding
of the possible definition and systematic description of
a catastrophic evolution scenario in
the language of quantum theory.

\subsection{Closed formulae}

Using a specific model we managed to cover
several methodical topics including not only the 
ambiguity of the metric
(i.e., of the necessary specification
of one of the eligible physical Hilbert spaces)
but also its descriptive aspects and appeal. Thus, we
paid attention
to the anisotropy of the alternative Hilbert-space geometries
as well as to the constraints imposed by the necessary 
positive-definiteness
of the mathematically correct
inner products in these spaces.

We formulated several arguments in favor
of our choice of
the toy model.
Firstly, it enabled us to employ the metodical idea of
paper \cite{lotor}. Thus, we played with the free
parameters in order to minimize the
anisotropy of the physical Hilbert-space metric $\Theta$.
Secondly, in the context of study of the EP degeneracy
as initiated in \cite{I} we found it useful to
invert the arrow of time.
Thus, in place of the time $t$ which starts at the EP
singularity
we used another time variable
$\tau$ which runs in opposite
direction. This enabled us to set the
initial zero long before
the fall of the system into its physical QC singularity.
Thirdly,
we found it productive to start our analysis from the systems
with the smallest level-multiplicities $N\leq 4$. Using the
brute-force linear algebra methods we managed to construct the fully
explicit matrices of the metrics which appeared (and were declared)
optimal and unique. 

Ultimately, the transparency and compact form of the
results of the brute-force linear-algebraic
calculations opened the way towards extrapolations.
Their use (followed by the decisively facilitated formal proofs)
finally enabled us to extend the validity of some of our
previous empirical small$-N$ observations to arbitrary
Hilbert-space dimensions $N$.

On this background, the
qualitative features of the QC process have been shown related to the
explicit qunatitative
properties of our manifestly time-dependent physical Hilbert-space metrics.
Thus, in a climax of the story
the picture of the
$N-$tuple QC level-degeneracy scenario
was given the form supported by the closed formulae
reflecting,
via the eigenvalues of the metric,
both the time-dependent anisotropy and
asymptotic degeneracy of the system in question.

\section{Summary\label{summary}}

In our present toy model
the evolution in time was explained as proceeding
from an innocent-looking and safely Hermitian
equidistant-energy-level onset prepared at an initial time
$\tau=0$
up to an ultimate collapse
realized via a complete, $N-$tuple EP
degeneracy of the energy spectrum at the final QC time $\tau=1$.

Our paper offered a compact and consistent picture
of the process in which the not quite expected exact
solvability of
our toy model enabled us to cover
all times $\tau\in (0,1)$. We were able to
describe the quantum-evolution fall
of the system in the
level-degeneracy quantum catastrophe,
and we were able to explain such a collapse
as a consequence of an unlimited growth
of the anisotropy of the underlying
time-dependent Hilbert space ${\cal H}^{(A)}$.

We found it natural to characterize
an optimal version of the latter
process by a minimal
spread of the set of eigenvalues $\theta^{(N)}_n$ of
the related
physical inner-product metric $\Theta^{(N)}$. We
decided to make such a metric unique via a
minimization of the latter
anisotropy-representing spread, with an emphasis put upon
the zero limit of the
special measure $\rho=\max (\theta^{(N)}_n-\theta^{(N)}_m)$
of the spread
at the onset $\tau=0$ of the process.

We managed to match our metric
smoothly to both of its extremes, i.e., not only to
the most common isotropic metric
at $\tau=0$
but also to the asymptotically degenerate metric at $\tau=1$.
In between these two extremes the
operator
(i.e., matrix)
has been kept smooth, nontrivial and optimal
during all $\tau \in (0,1)$.
From the point of view of phenomenology, we arrived at a benchmark
quantum representation of the EP-related catastrophe
in which the fall into the
degeneracy appeared realized
in finite time.

Our toy-model simulation of the catastrophe
can be perceived as initiated by
an arbitrary conventional unitary-evolution
prehistory at $\tau<0$.
According to the general principles of quantum theory
the states of the system during its EP-related degeneracy
at $\tau \in (0,1)$ were assumed described differently, by
a THS wave function $\psi$ which evolves in time in a way which has 
thoroughly been described in Ref.~\cite{SIGMA}. 
In our present paper
we skipped most of the related technical details
and we
restricted our attention just to the
description of the interplay of
time-dependence between a pre-selected
``non-Hermitian'' benchmark Hamiltonians
$H^{(N)}(\tau)$
and one of the eligible ``Hermitizing''
metric operators $\Theta^{(N)}(\tau)$.

Our choice of the latter
operator can be characterized
as truly exceptional:
In the context of mathematics
its form
has been shown to lead
to a unique, minimally anisotropic
geometry of the physical Hilbert space of states.
In the context of physics
we emphasized that
a decisive merit of
the time-dependence of
our inner-product metric
should be seen in the guarantee of existence
of this metric
up to an arbitrarily small vicinity of the
ultimate catastrophic quantum collapse of the system.

%


\end{document}